\documentclass[pre,showkeys]{revtex4}
\usepackage[dvips]{graphicx}
\usepackage{amsthm,amssymb,bm}

\newcommand{\beq}[0]{\begin{equation}}
\newcommand{\eeq}[0]{\end{equation}}
\newcommand{\e}{\varepsilon}
\newcommand{\thet}{\vartheta}
\newcommand{\la}{\langle}
\newcommand{\ra}{\rangle}
\newcommand{\ds}{\displaystyle}
\newcommand{\avh}{\la H_1 \ra}
\newcommand{\pa}{\partial}
\newcommand{\w}{\omega}
\newcommand{\ud}{{\mathrm d}}
\newcommand{\bc}{\begin{center}}
\newcommand{\ec}{\end{center}}

\theoremstyle{plain} \newtheorem{theorem}{Theorem}
\theoremstyle{plain} \newtheorem{lemma}{Lemma}

\begin{document}

\title{Non-Existence of phase-shift breathers in one-dimensional Klein-Gordon lattices with nearest-neighbor interactions}

\author{Vassilis Koukouloyannis}

\affiliation{Department of Physics\\
Section of Astrophysics, Astronomy and Mechanics\\ 
Aristotle University of Thessaloniki\\ 
GR-54124 Thessaloniki, Greece}


\begin{abstract}
It is well known that one-dimensional Klein-Gordon lattices with nearest-neighbor interactions can support multibreathers with phase differences between the successive ``central'' oscillators $\phi_i=0\ \mbox{or}\ \pi$ (standard configurations). In this paper we prove that in this kind of systems, the standard configurations are the only possible ones, so phase-shift breathers (configurations with $\phi_i\neq0,\, \pi$) cannot be supported. This fact also determines the linear stability of the existing multibreathers.
\end{abstract}

\maketitle

\section{Introduction}

Since~\cite{sietak,pa90} much interest has been drawn in the study of space-localized and time-periodic motions in lattices of coupled oscillators. These motions are called {\it discrete breathers} (DB) if the oscillation is localized around one ``central'' lattice site, while, if there are more than one central oscillators, the motion is called {\it multibreather} (MB) or {\it multi-site breather}. The wide interest about discrete breathers- miltibreathers is underlined by the numerous review papers there exist on this subject (e.g. ~\cite{flach1,macrev,flach2,aubrev}).

One of the most popular systems in which such motions are studied is the well-known Klein-Gordon (KG) chain. The classical KG setting consists of a one-dimensional lattice of oscillators each coupled with its nearest-neighbors (NN). Since the first proof of existence of DBs \cite{macaub}, there have been several papers dealing with the issue of existence and stability of MBs in KG chains (e.g. \cite{ahn1, sepmac,flachproof}). In \cite{koukmac} a methodology for proving the existence of multi-site breathers was introduced based in the work of \cite{mackay1,mackay2} and using also the terminology of \cite{koukicht1}. This methodology was generalized for a generic Klein-Gordon chain in \cite{KK09} and provided general persistence and stability conditions independently of the precise form of the on-site potential. The stability results of \cite{KK09}, have been shown in \cite{IJBC} to be in correspondence to the results of \cite{ACSA03}, which were already been obtained by using the band theory of \cite{aubrev}. In a recent work \cite{Rapti13lri} an alternative proof of the stability theorem of \cite{ACSA03} has been presented by using the same band theory, providing also a proof of the eigenvalue counting result of \cite{ACSA03}. These results have been generalized in \cite{pelisak} by also considering ``holes'' between the central oscillators. In a different context, similar results have been recently obtained \cite{yoshi} by considering a diatomic FPU chain.

The existing multibreather solutions are categorized in terms of the phase differences between the central oscillators. It is well known that KG chains can support multibreathers with phase differences between the successive central oscillators $\phi_i=0, \pi$. These are the {\it standard configurations}. Although there is strong evidence that {\it phase-shift breathers} i.e.\,multibreathers with phase differences $\phi_i\neq0, \pi$, cannot exist, a rigorous proof of this fact had not been presented. 

In this work we prove that the one-dimensional Klein-Gordon lattice with nearest-neighbor interactions cannot support phase-shift breathers, by proving that the persistence conditions provided by \cite{KK09} do not have solutions other than the standard ones $\phi_i=0,\pi$. 

The fact of the non-existence of phase-shift breathers in KG chains determines also the stability of the standard configurations. In particular, in \cite{KK09} the main theorem has been stated under the assumption of non-existence of phase-shift configurations. On the other hand, if we consider a KG chain with interactions between its oscillators further than mere the nearest-neighbors, phase-shift breathers {\it can} be supported \cite{lri} and consequently the stability picture radically changes.

The paper is organized as follows; in section \ref{summary} we present briefly the methodology for the derivation of the persistence conditions for the existence of multibreathers in KG chains developed in \cite{KK09}, while we introduce some terminology. In section \ref{nonexistence} the main theorem about the non-existence of phase-shift breathers is proven. In section \ref{stability} we discuss the implication of this theorem to the stability of the standard MB configurations.

\section{Persistence and stability of multibreathers in 1D Klein-Gordon lattices with nearest-neighbor interactions}\label{summary}
In this section we will shortly present the main results of \cite{KK09}, concerning the existence of multibreathers (MB) in a Klein-Gordon (KG) chain. The classical KG setting is defined as a 1D lattice of coupled oscillators each one moving in a nonlinear potential $V(x)$ possessing a local minimum at $x=0$ ($V'(0)=0, V''(0)=\w_p^2>0)$. Each oscillator is coupled with its two nearest neighbors (NN) with a linear coupling force through a coupling constant $\e$, as shown in Fig.\ref{fig:KGchain}.
\begin{figure}[h]
	\centering
		\includegraphics[width=8cm]{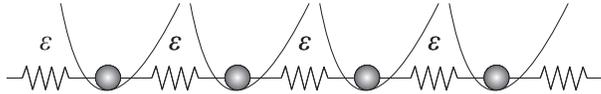}
		\caption{A one-dimensional Klein-Gordon lattice with nearest-neighbor interactions}
	\label{fig:KGchain}
\end{figure}\\
The Hamiltonian of a Klein-Gordon chain with nearest neighbor interactions is the following
\beq H=H_0+\e H_1=\sum_{i=-\infty}^{\infty}\left[ \frac{1}{2}p_i^2 +V(x_i)\right] +\frac{\e}{2}\sum_{i=-\infty}^{\infty}\left(x_i-x_{i-1}\right)^2,\label{ckg}\eeq
which leads to the equations of motion
$$\ddot{x_i}=-V'(x_i)+\e(x_{i-1}-2x_i+x_{i+1}).$$
This system is well known to support discrete breather, as well as, multibreather solutions (e.g. \cite{macaub,ahn1, mackay1, KK09}).

The key notion to the proof of the existence of multibrethears is this of the anticontinuum limit. This is the limit $\e\rightarrow0$ where the chain consists of uncoupled oscillators. In this limit we consider all the oscillators of the chain at rest except for $n+1$ adjacent ``central'' ones which move in periodic orbits of frequency $\w$, but with arbitrary phases. This configuration defines a trivially space-localized and time-periodic motion. But, not all of these configurations survive when coupling is introduced ($\e\neq0$) to provide a multibreather. In order for these motions to persist for $\e\neq0$, specific conditions on the phase differences between the oscillators must be satisfied, as well as, some rather generic non-degeneracy conditions.

In \cite{mackay1} it was shown that multibreathers
correspond to critical points of $H^{\mathrm{eff}}$ which to leading order of
approximation is given by $H^{\mathrm{eff}}=H_0(I_i)+\e\avh(\phi_i, I_i)$ \cite{koukmac}. The variables $\phi_i=w_{i+1}-w_i$ denote the $n$ phase differences of the $n+1$ successive central oscillators, while $I_i$ are given by $I_i=\sum_{j=i}^{n} J_j$, where $(J_i, w_i)$ are the action-angle variables of a single oscillator. Note that, if we had considered also ``holes'' between the central oscillators, i.e.\. oscillators between the central ones which in the anticontinuum limit lie at rest, a higher order approximation of $H^{\mbox{eff}}$ should be used.

The average value of the coupling part of the Hamiltonian
$$\avh(\phi_i, I_i)=\frac{1}{T}\oint H_1(w_0, \phi_i, I_i)\ud t$$
is calculated along the orbits in the anti-continuum limit $\e=0$.

This yields the conclusion that the persistence conditions for the existence of $n+1$-site multibreathers are
\beq\frac{\pa \avh}{\pa \phi_i}=0,\quad i=1\ldots n,\label{gen_per}\eeq
as far as two non-degeneragy conditions hold. The first one is the non-resonance condition of the breather frequency $\w$ with the phonon frequency $\w_p$ i.e. $\w_p\neq k\w$. The second condition is  the anharmonicity condition $\frac{\pa \w}{\pa J}\neq0$ which implies that the oscillation frequency of a single oscillator depends on the oscillation amplitude.

By using the fact that the motion of the central oscillators for $\e=0$ can be described by
\beq x_i=\sum_{m=0}^\infty A_m\cos(mw_i)\label{xdevel}\eeq
the average value of $H_1$ becomes (\cite{KK09} appendix A)
$$\avh=-\frac{1}{2}\sum_{m=1}^\infty\sum_{s=1}^{n}A_m^2\cos(m\phi_s)$$
and the persistence conditions (\ref{gen_per}) become in the case of Klein-Gordon chains with nearest neighbor interactions,
\beq\frac{\pa \avh}{\pa \phi_i}=0\Rightarrow M(\phi_i)\equiv\sum_{m=1}^{\infty}mA_m^2\sin(m\phi_i)=0,\quad i=1\ldots n.\label{per_r1}\eeq
This system of equations possesses the obvious solutions $\phi_i=0, \pi$, while, as it will be shown in the next 
section it possesses no others.

Since only $\phi_i=0\ \mbox{or}\ \pi$ can be supported, the multibreathers should have any pair of adjacent central oscillators moving either in-phase or anti-phase. These configurations are called the {\it standard configurations}. In figure \ref{3site} all the possible standard 3-site breather configurations are shown. These are (a) the in-phase \{$\phi_1=\phi_2=0$\} configuration, (b) the anti-phase \{$\phi_1=\phi_2=\pi$\} configuration and (c) the mixed one \{$\phi_1=0,\ \phi_2=\pi$\}. The on-site potential used in order to acquire these figures is $V(x)=x^2/2-0.15x^3/3-0.05x^4/4$ and the depicted multibreathers correspond to coupling $\e=0.02$ and frequency $\w=0.845$.  

\begin{figure}[ht]
\begin{tabular}{ccc}
\includegraphics[width=5cm]{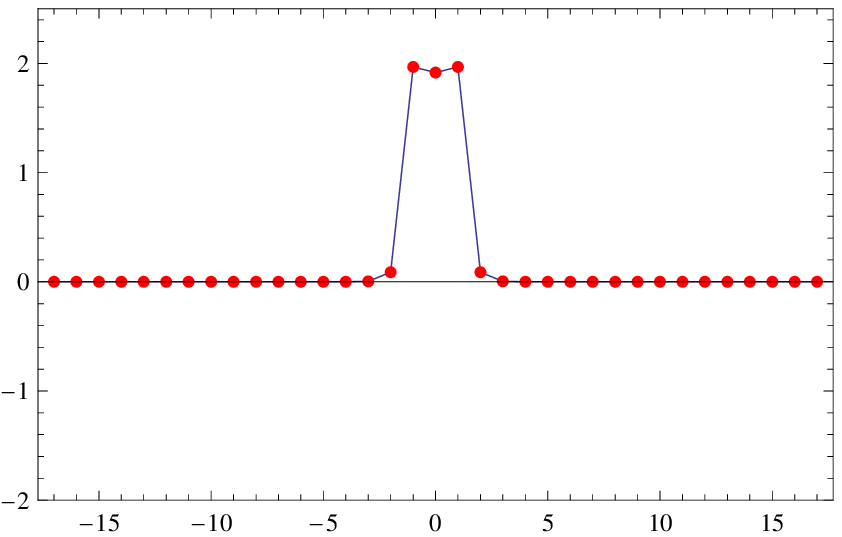}&\includegraphics[width=5cm]{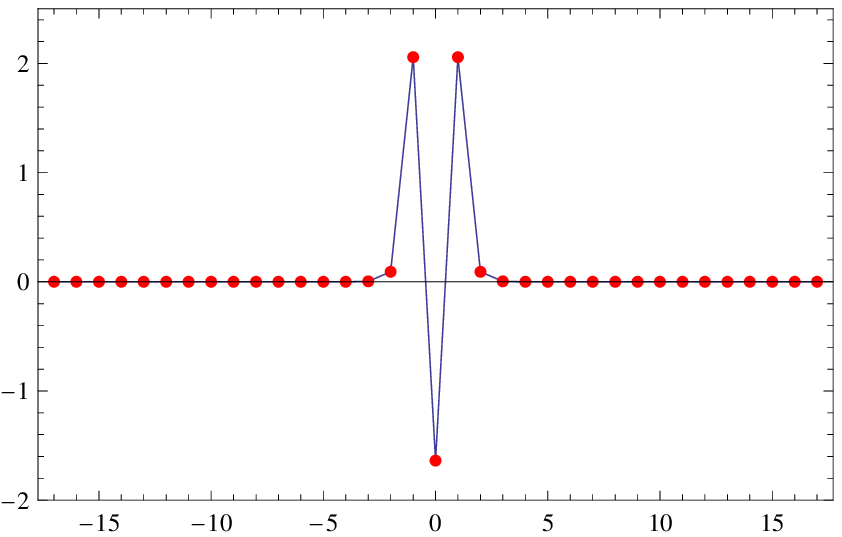}&\includegraphics[width=5cm]{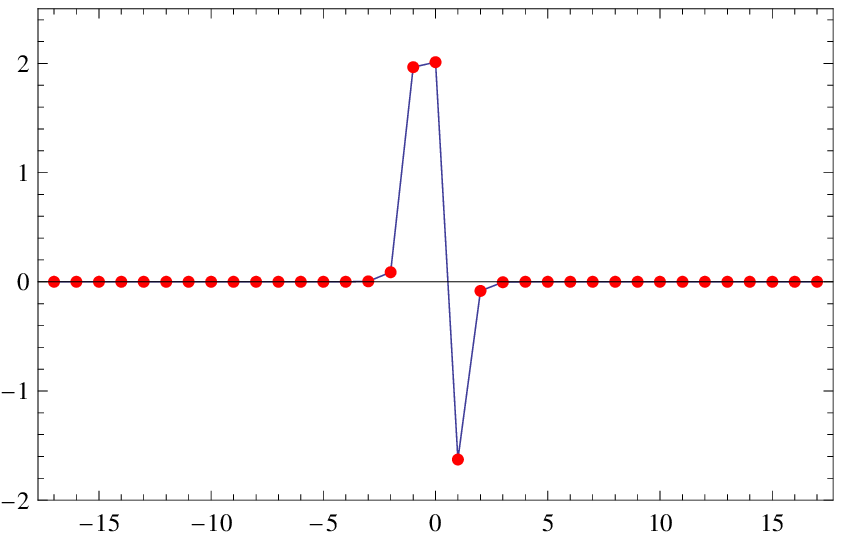}\\
(a)&(b)&(c)
\end{tabular}
\caption{Snapshots of all the possible 3-site breather configurations. In (a) the in-phase \{$\phi_1=\phi_2=0$\} configuration is shown, in (b) the anti-phase \{$\phi_1=\phi_2=\pi$\} and in (c) the mixed one \{$\phi_1=0,\ \phi_2=\pi$\}.}
\label{3site}
\end{figure}
In \cite{lri} it has been shown that, by introducing interactions with range beyond this of just the nearest-neighbor ones, the persistence condition (\ref{per_r1}) are altered, so {\it phase-shift configurations} ($\phi_i\neq0\ \mbox{or}\ \pi$) are also permitted. One could intuitively predict the existence of phase-shift configurations, since a 1D chain with long-range interactions can be analogous to a nearest-neighbor interaction 2D lattice which supports such configurations. For example, a three-site breather in a next-nearest-neighbor 1D lattice is equivalent to a three-site (triangular) breather in a nearest-neighbor hexagonal lattice, where the existence of a vortex-breather configuration with $\phi_1=\phi_2=2\pi/3\ \mbox{or}\ 4\pi/3$ is well established \cite{koukmac}.

\section{Proof of nonexistence of phase-shift breathers in one-dimensional Klein-Gordon chains with nearest-neighbor interactions}\label{nonexistence}
In this section we will prove that phase-shift breathers\footnote{Note the use og the term ``phase-shift breathers'' instead of ``phase-shift multibreathers''. This is done for simplicity, since the phase-shift breathers are by definition multi-site being characterized by the phase difference between the adjacent central oscillators.} ($\phi_i\neq0\, \mbox{or}\, \pi$) cannot be supported in one-dimensional Klein-Gordon chains with nearest-neighbor interactions. This fact will be proven by showing that the persistence conditions (\ref{per_r1}) have only the $\phi_i=0,\pi$ solutions in the $\phi_i\ \in\ [0, 2\pi)$ interval. 

In order to prove our main theorem we have first to prove two lemmas.

\begin{lemma}\label{lem1}
The solutions of $M(\phi)=0$ coincide with the solutions of $\ds I(\phi)\equiv\int_0^{2\pi}\dot{x}(w)~x(w-\phi) \ud w = 0$.
\end{lemma}
\begin{proof}
As we have already mentioned, the displacement of an uncoupled oscillator from the equilibrium can be described as an even $2\pi$-periodic function with respect to $w=\w t+\thet$, as
\beq x(w)=\sum_{m=0}^\infty A_m\cos(mw).\label{eq:x}\eeq
We define the function $N(\phi)$ as the opposite of the averaged autocorrelation function of $x(w)$
\beq N(\phi)=-\frac{1}{2\pi}\int_0^{2\pi}x(w)x(\phi-w)\ud w.\label{eq:m}\eeq
By substituting (\ref{eq:x}) into (\ref{eq:m}) we get
$$\ds N(\phi)=-\frac{1}{2\pi}\sum_{n=0}^\infty\sum_{m=0}^\infty A_nA_m\int_0^{2\pi}\cos(nw)\cos[m(\phi-w)]\ud w$$
or
$$\hspace{-0.5cm} N(\phi)=-\frac{1}{4\pi}\sum_{n=0}^\infty\sum_{m=0}^\infty A_nA_m\int_0^{2\pi}\{\cos[(n-m)w+m\phi]+\cos[(n+m)w-m\phi]\}\ud w$$
which leads to
\beq N(\phi)=-\frac{1}{2}\sum_{m=0}^\infty A_m^2\cos(m\phi).\label{eq:m2}\eeq
By differentiation of (\ref{eq:m2}) with respect to $\phi$ we get
\beq\frac{\ud}{\ud \phi}N(\phi)=\frac{1}{2}\sum_{m=0}^\infty mA_m^2\sin(m\phi)=\frac{1}{2}M(\phi).\label{dn2}\eeq
On the other hand, by using the differentiation properties of the convolution function, we get from (\ref{eq:m}),
\beq\frac{\ud}{\ud \phi}N(\phi)=-\frac{1}{2\pi}\int_0^{2\pi}\frac{\ud x(w)}{\ud w}~x(\phi-w)\ud w=-\frac{1}{2\pi\w}\int_0^{2\pi} \dot{x}(w)~x(\phi-w)\ud w,\label{dn1}\eeq
where the dot denotes differentiation with respect to time. From (\ref{dn2}) and (\ref{dn1}) we get finally
\beq M(\phi)\equiv\sum_{m=0}^\infty mA_m^2\sin(m\phi)=-\frac{1}{\pi\w}\int_0^{2\pi}\dot{x}(w)~x(\phi-w)\ud w=-\frac{1}{\pi\w}\int_0^{2\pi}\dot{x}(w)~x(w-\phi)\ud w=-\frac{1}{\pi\w}I(\phi).\label{eq:eq}\eeq
So, 
\beq M(\phi)=0\Leftrightarrow I(\phi)=0\label{eq:MequivI}\eeq
\end{proof}

\noindent {\bf Remark:} 
By the differentiation properties of the convolution function, we could get also from (\ref{eq:m}),
$$N'(\phi)=-\frac{1}{2\pi}\int_0^{2\pi}x'(w)~x(\phi-w)\ud w
=-\frac{1}{2\pi}\int_0^{2\pi} x(w)~x'(\phi-w)\ud w,$$
where the prime denotes differentiation with respect to the argument of the function, or
\beq\int_0^{2\pi}\dot{x}(w)~x(x-\phi)\ \ud w=-\int_0^{2\pi}x(w)~\dot{x}(w-\phi)\ \ud w.\label{eq:oposite}\eeq
In the last equation we have used the symmetries of the displacement $x(w)$, i.e. $x(-w)=x(w)$ and $\dot{x}(-w)=-\dot{x}(w)$. This last fact will be used in the proof of lemma \ref{lem2}.
\begin{lemma}\label{lem2}
Equation $I(\phi)=0$ has no solutions for $\phi \in (0,\ \pi)\cup (\pi,\ 2\pi)$.
\end{lemma}
\begin{proof}

Let ${\bf x}(w)$ be the vector ${\bf x}(w)=\{x(w), \dot{x}(w)\}$. So, the cross product of ${\bf x}(w)$ with ${\bf x}(w-\phi)$ reads\footnote{Here, we have considered the $\left|\bullet\right|$ notation for the signed value of the cross product and the $\left\|\bullet\right\|$ notation for the norm of the vector.}

$$|{\bf x}(w)\times{\bf x}(w-\phi)|=x(w)\dot{x}(w-\phi)-\dot{x}(w)x(w-\phi)=\left\|{\bf x}(w)\right\| \left\|{\bf x}(w-\phi)\right\|\sin\phi$$
So,
$$\int_0^{2\pi}[x(w)\dot{x}(w-\phi)-\dot{x}(w)x(w-\phi)]\ud w=\sin\phi\int_0^{2\pi}\left\|{\bf x}(w)\right\|\left\|{\bf x}(w-\phi)\right\|\ud w$$
or, by using (\ref{eq:oposite}),
$$-2\int_0^{2\pi}\dot{x}(w)~x(w-\phi)~\ud w=\sin\phi\int_0^{2\pi}\left\|{\bf x}(w)\right\|\left\|{\bf x}(w-\phi)\right\|~\ud w$$
and finally
\beq I(\phi)\equiv\int_0^{2\pi}\dot{x}(w)~x(w-\phi)\ud w=-\frac{\sin\phi}{2}\int_0^{2\pi}\left\|{\bf x}(w)\right\|\left\|{\bf x}(w-\phi)\right\|\ud w.\label{eq:I_lem}\eeq
Since $\left\|x(w)\right\|>0$ $\forall~w$, for the unperturbed oscillator, the sign of $I(\phi)$ is defined by the value of $\phi$. 
So, $I(\phi)<0$ for $0<\phi<\pi$ and $I(\phi)>0$ for $\pi<\phi<2\pi$. 
\end{proof}

\begin{theorem}\label{thm2}
In a one-dimensional Klein-Gordon lattice with nearest-neighbor interactions only standard configuration ($\phi_i=0\ \mbox{or}\ \pi$) multibreathers can be supported.
\end{theorem}

\begin{proof}
In order for a MB configuration to be supported in a classical Klein-Gordon chain described by (\ref{ckg}), the phase-differences $\phi_i$ between successive oscillators must satisfy eqs.(\ref{per_r1}).

The system (\ref{per_r1}) is decomposed into $n$ independent equations, so it is sufficient to check for solutions of $M(\phi)=0$ and since $M(\phi)$ is $2\pi$-periodic it is sufficient to check for solutions with $\phi \in [0,\ 2\pi)$.

Since $\phi=0,\, \pi$ are by construction solutions of $M(\phi)=0$, the standard configuration multibreathers are always supported.

On the other hand, by (\ref{eq:eq}) and (\ref{eq:I_lem}), we get that

\beq M(\phi)=\frac{\sin\phi}{2\pi\w}\int_0^{2\pi}\left\|{\bf x}(w)\right\|\left\|{\bf x}(w-\phi)\right\|\ud w.\label{eq:M_final}\eeq
So,
$$M(\phi)\neq0\ \mbox{for}\ \phi \in (0,\ \pi)\cup(\pi,\ 2\pi),$$
and phase-shift ($\phi\neq0\ \mbox{or}\ \pi$) configurations cannot be supported.

In addition, from (\ref{eq:M_final}) we get that $M(\phi)>0$ for $0<\phi<\pi$ and $M(\phi)<0$ for $\pi<\phi<2\pi$.
\end{proof}

\noindent{\bf Remark:} One could think the case of phonobreathers (see e.g. \cite{CAR11}) as a case where phase-shift configurations exist in a Klein-Gordon chain with nearest-neighbor interactions. But, these motions are substantially different from the
multibreathers we study in this work, since in the case of phonobreathers all the sites of the lattice are excited in the anticontinuoum limit, while in our case
there is a specific number of $n+1$ adjacent central oscillators. This fact together with the assumption of periodic boundary conditions give rise to different persistence conditions than (\ref{per_r1}) which support phase-shift configurations.

\section{Discussion about the stability of multibreathers in Klein-Gordon chains}\label{stability}
In \cite{KK09}, under the assumption of non-existence of phase-shift breathers, the spectral stability of the multibreather solutions in one-dimensional Klein-Gordon chains with nearest-neighbor interactions, is well established by the theorem
\begin{theorem}\label{thm1}
Under the assumption that (\ref{per_r1}) has only the $\phi_i=0,\pi$ solutions, in systems of the form (\ref{ckg}), if $P\equiv\e\frac{\pa \w}{\pa J}<0$ the only configuration which leads to
linearly stable multibreathers, for $|\e|$ small enough, is the one
with $\phi_i=\pi\quad\forall i=1\ldots n$ (anti-phase
multibreather), while if $P>0$ the only linearly
stable configuration, for $|\e|$ small enough, is the one with
$\phi_i=0\quad\forall i=1\ldots n$ (in-phase multibreather).
Moreover, for $P<0$ (respectively, $P>0$), for unstable configurations, their number of
unstable eigenvalues will be precisely equal to the number of nearest
neighbors which are in- (respectively, in anti-) phase between them.
\end{theorem}

After the proof of theorem \ref{thm2} the assumption is no longer necessary.  
So, from the above, we conclude that the non-existence of phase-shift breathers is significant, not only in order to exclude the non-supported multibreather configurations in 1D KG chains but also in order to categorize the supported ones in terms of their corresponding stability.

It is important to bare in mind that the above theorem holds for $|\e|$ close to the anticontinuum limit, since as the coupling strength increases the stability of a specific configuration can change through a Hamiltonian-Hopf bifurcation. On the other hand the result of Theorem \ref{thm2} on the existence (or non-existence) of the discussed configurations holds independently of this change of the stability. 

We have to note here that in a recent paper \cite{pelisak} the authors have studied the stability of configurations with holes between the excited oscillators in the anticontinuum limit, by using higher order perturbation theory. But, in this work we will only consider adjacent central oscillators.

\section{Conclusions}
It is well known that one-dimensional Klein-Gordon (KG) lattices with nearest-neighbor (NN) interactions support multibreathers  with the standard phase-difference $\phi_i=0, \pi$ between adjacent central oscillators. On the other hand there were strong evidences  (including numerical computations) suggesting that phase-shift breathers i.e. multibreathers with $\phi_i\neq0\mbox{ or }\pi$ cannot exist in this classical KG setting.

In the present work we prove that, indeed, the only configurations that can exist in a classical KG 1D lattice with NN interactions are the standard ones ($\phi_i=0,\pi$). This fact excludes the existence of phase-shift breathers and, as it has been shown in \cite{KK09}, it also clarifies the stability image for the existing multibreathers i.e. if
$P\equiv\e\frac{\pa \w}{\pa J}<0$ the anti-phase configuration is the only stable
one, while for $P>0$ the in-phase configuration is
the only stable multibreather solution.

On the other hand, as it has been recently shown \cite{lri}, in 1D KG chains where interactions with range larger than just the nearest-neighbor ones are considered, phase-shift breathers can be supported, giving rise to radically different stability scenaria.

Future directions of this work could include the study of the possibility of use of similar techniques in order to infer supported solutions in the case of lattices with longer range interactions or in the case of multibreathers with holes between the central oscillators.

\section*{Acknowledgements}
This research has been co-financed by the European Union (European Social Fund - ESF) and Greek national funds through the Operational Program "Education and Lifelong Learning" of the National Strategic Reference Framework (NSRF) - Research Funding Program: THALES. Investing in knowledge society through the European Social Fund.

The author would like to thank Prof. D.E. Pelinovsky for a conversation which rekindled the interest about this work  and Prof. P.G. Kevrekidis for his invaluable help through many fruitful discussions and continuous encouragement. The author would like to thank also the anonymous referee who indicated that the previous version of the proof was unnecessarily long. 

\bibliographystyle{plain}
\bibliography{non}

\end{document}